\documentclass[a4paper,UKenglish,cleveref, autoref, thm-restate]{lipics-v2021}

\usepackage{mathtools}
\usepackage{listings}
\usepackage{xspace}
\usepackage{etoolbox}
\usepackage{tikz}
\usetikzlibrary{arrows.meta, cd}

\pdfoutput=1
\hideLIPIcs

\nolinenumbers{}

\newcommand{\code}[1]{\lstinline!#1!}

\definecolor{egreen}{rgb}{0, 0.4, 0.267}
\definecolor{dkviolet}{rgb}{0.6,0,0.8}
\definecolor{dkgreen}{rgb}{0,0.4,0}
\definecolor{dkblue}{rgb}{0,0.1,0.5}
\definecolor{lightblue}{rgb}{0,0.5,0.5}
\definecolor{orange}{rgb}{0.9,0.39,0}

\newcommand{\namefont}[1]{\textsf{#1}}
\newcommand{\Coq}{\namefont{Coq}\xspace}
\newcommand{\Lean}{\namefont{Lean}\xspace}
\newcommand{\Globular}{\namefont{Globular}\xspace}

\newcommand{\ie}{i.e.}
\newcommand{\eg}{e.g.}
\newcommand{\resp}{resp.\ }

\renewcommand{\~}{\widetilde}
\renewcommand{\emptyset}{\varnothing}

\newcommand{\N}{\mathbb N}

\newcommand{\Cat}{\mathcal C}
\newcommand{\T}{\mathcal T}
\newcommand{\M}{\mathcal M}
\newcommand{\A}{\mathcal A}
\newcommand{\Q}{\mathcal Q}
\newcommand{\BP}{\mathcal{BP}}

\newcommand{\x}{\mathrm{x}}
\newcommand{\PQ}{\mathrm{PQ}}
\newcommand{\im}{\mathrm{Im}}
\newcommand{\card}{\mathrm{card}}
\newcommand{\Ob}{\mathrm{Ob}}
\newcommand{\Hom}{\mathrm{Hom}}
\newcommand{\restr}{\mathrm{restr}}
\newcommand{\com}{\mathrm{commute}}
\newcommand{\id}{\mathrm{id}}

\newcommand{\cone}{\mathrm{cone}}
\newcommand{\comp}{\mathrm{comp}}

\newcommand{\cat}{\mathrm{cat}}
\newcommand{\ab}{\mathrm{ab}}

\newcommand{\sop}{\textsf{sp}}
\newcommand{\tp}{\textsf{tp}}
\newcommand{\sot}{\textsf{st}}
\newcommand{\tot}{\mathrm{tot}}

\newcommand{\free}[1]{{\langle #1 \rangle}}
\newcommand{\rel}[1]{\sim}
\newcommand{\eqd}{\approx}
\newcommand{\cyc}{\mathring}
\newcommand{\dual}{\dagger}

\newcommand{\RestrComp}{\mathrm{RestrComp}}

\newcommand{\Pushout}{\mathrm{Cospan}}
\newcommand{\PushoutE}{\mathrm{PushoutEU}}
\newcommand{\Comp}{\mathrm{Comp}}
\newcommand{\CompE}{\mathrm{CompE}}
\newcommand{\ComEq}{\mathrm{ComEq}}

\newcommand{\EqPath}{\mathrm{EqPath}}
\newcommand{\EqPathRefl}{\mathrm{EqPathRefl}}
\newcommand{\EqPathSym}{\mathrm{EqPathSym}}
\newcommand{\EqPathTrans}{\mathrm{EqPathTrans}}
\newcommand{\PathCom}{\mathrm{PathCom}}
\newcommand{\EqPathConcat}{\mathrm{EqPathConcat}}
\newcommand{\Id}{\mathrm{Id}}
\newcommand{\IdE}{\mathrm{IdE}}
\newcommand{\EmptyEU}{\mathrm{EmptyEU}}
\newcommand{\CommergeWithId}{\mathrm{CommergeWithId}}

\newcommand{\Mono}{\mathrm{Mono}}
\newcommand{\Epi}{\mathrm{Epi}}
\newcommand{\Cone}{\mathrm{Cone}}
\newcommand{\Limit}{\mathrm{Limit}}
\newcommand{\Colimit}{\mathrm{Colimit}}
\newcommand{\Zero}{\mathrm{Zero}}
\newcommand{\ZeroE}{\mathrm{ZeroE}}
\newcommand{\Ker}{\mathrm{Ker}}
\newcommand{\Coker}{\mathrm{Coker}}
\newcommand{\ProductE}{\mathrm{ProductE}}
\newcommand{\CoproductE}{\mathrm{CoproductE}}
\newcommand{\KerE}{\mathrm{KerE}}
\newcommand{\CokerE}{\mathrm{CokerE}}
\newcommand{\MonoNormal}{\mathrm{MonoNormal}}
\newcommand{\EpiNormal}{\mathrm{EpiNormal}}
\newcommand{\Commerge}{\mathrm{Commerge}}

\NewDocumentCommand{\quiver}{ O{1} m O{0} m O{0} }{
  \tikz[baseline={([yshift=\ifblank{#1}{-.6ex}{0ex}]current bounding box.center)}, inner sep=\ifblank{#1}{1pt}{.5pt}, -{Latex[scale=\ifblank{#1}{1}{.5}]}, scale=\ifblank{#1}{1}{.5}]{
    \bfseries
    \foreach \p [count = \j, evaluate={\i=int(\j-1)}, evaluate={\k=int(100-(\i<#3)*70)}] in {#2} {
      \node[black!\k] (\i) at \p {.};
    }
    \foreach \i/\j [count = \l, evaluate={\k=int(100-(\l<=#5)*70)}] in {#4} {
      \draw[black!\k] (\i)--(\j);
    }
  }
}

\NewDocumentCommand{\quiverr}{ O{1} m O{0} m O{0} }{
  \tikz[baseline={([yshift=\ifblank{#1}{-.5ex}{0ex}]current bounding box.center)}, inner sep=\ifblank{#1}{1pt}{.5pt}, -{Latex[scale=\ifblank{#1}{1}{.5}]}, scale=\ifblank{#1}{1}{.5}]{
    \bfseries
    \foreach \p [count = \j, evaluate={\i=int(\j-1)}, evaluate={\k=int(100-(\i<#3)*70)}] in {#2} {
      \node[black!\k] (\i) at \p {.};
    }
    \foreach \i/\j/\s [count = \l, evaluate={\k=int(100-(\l<=#5)*70)}] in {#4} {
      \draw[black!\k] (\i)to[bend left=\s](\j);
    }
  }
}

\newcommand{\Qdot}[1][1]{{\quiver[#1]{(0,0)}{}}}\newcommand{\Qtwodots}[1][1]{{\quiver[#1]{(0,0),(.2,0)}{}}}
\newcommand{\Qloop}{\tikz[baseline={(current bounding box.center)}, inner sep=.5pt, scale=.5]{\node (0) at (0,0) {\textbf{.}}; \draw[-{Latex[scale=.5]}] (0) to[out=120, in=60, looseness=10] (0);}}
\newcommand{\Qmap}[1][1]{{\quiver[#1]{(0,0),(.8,0)}{0/1}}}
\newcommand{\Qbimap}[1][1]{{\quiver[#1]{(0,0),(1,0),(2,0)}{0/1,1/2}}}
\newcommand{\Qcobimap}[1][1]{{\quiverr[#1]{(0,0),(1,0)}{0/1/30,0/1/-30}}}
\newcommand{\Qcomp}[1][1]{{\quiver[#1]{(0,0),(.5,-.5),(1,0)}{0/1,0/2,1/2}}}
\newcommand{\Qmono}[1][1]{{\quiverr[#1]{(0,0),(.8,0),(1.6,0)}{0/1/20,0/1/-20,0/2/35,1/2/0}}}
\newcommand{\Qker}[1][1]{{\quiver[#1]{(0,0),(.7,.2),(.9,-.2),(1.6,0)}{0/1,0/2,0/3,1/3,2/3}}}

\newcommand{\mapicomp}{{\quiver{(1,0),(.5,-.5),(0,0)}[1]{1/0,2/0,2/1}[2]}}
\newcommand{\mapiocomp}{{\quiver{(.5,-.5),(0,-0),(1,0)}[1]{1/0,0/2,1/2}[2]}}
\newcommand{\mapiicomp}{{\quiver{(0,0),(.5,-.5),(1,0)}[1]{0/1,0/2,1/2}[2]}}
\newcommand{\dotimap}{{\quiver{(1,0),(0,0)}[1]{1/0}[1]}}
\newcommand{\dotiomap}{{\quiver{(0,0),(1,0)}[1]{0/1}[1]}}
\newcommand{\twodotsimap}[1][1]{{\quiver[#1]{(0,0),(1,0)}{0/1}[1]}}
\newcommand{\mapioomono}  {{\quiverr{(0,0),(.8,0),(1.6,0)}[1]{0/1/20,0/1/-20,0/2/35,1/2/0}[3]}}
\newcommand{\compimono}   {{\quiverr{(0,0),(.8,0),(1.6,0)}{0/1/-20,0/1/20,0/2/35,1/2/0}[1]}}
\newcommand{\compiomono}  {{\quiverr{(0,0),(.8,0),(1.6,0)}{0/1/20,0/1/-20,0/2/35,1/2/0}[1]}}
\newcommand{\cobimapimono}{{\quiverr{(1.6,0),(.8,0),(0,0)}[1]{2/0/35,1/0/0,2/1/20,2/1/-20}[2]}}
\newcommand{\dotiiker}{{\quiver{(0,0),(.7,.2),(1.6,0),(.9,-.2)}[3]{0/1,0/3,0/2,1/2,3/2}[5]}}
\newcommand{\mapiooker}{{\quiver{(0,0),(.9,-.2),(.7,.2),(1.6,0)}[2]{0/2,0/1,0/3,1/3,2/3}[4]}}
\newcommand{\mapiker}{{\quiver{(.9,-.2),(1.6,0),(0,0),(.7,.2)}[2]{2/0,2/1,3/1,0/1,2/3}[4]}}
\newcommand{\veeiker}{{\quiver{(0,0),(.7,.2),(.9,-.2),(1.6,0)}[1]{0/1,0/2,0/3,1/3,2/3}[3]}}
\newcommand{\mapibimap}{\quiver{(2,0),(1,0),(0,0)}[1]{1/0,2/1}[1]}
\newcommand{\mapiobimap}{\quiver{(0,0),(1,0),(2,0)}[1]{0/1,1/2}[1]}

\newcommand{\linebox}{\path[use as bounding box] (-.05,0ex)--(.55,0ex)}
\newcommand{\spo} {\tikz[scale=.5, thick]{\linebox; \draw[black!30] (0,0) node[black] {\textbf{.}} to[out=40,in=-140] (.5,0);}}
\newcommand{\tpo} {\tikz[scale=.5, thick]{\linebox; \draw[black!30] (0,0) to[out=40,in=-140] (.5,0) node[black] {\textbf{.}};}}
\newcommand{\spd} {\tikz[scale=.5, thick]{\linebox; \draw[black!30] (0,0) to[out=40,in=-140] (.5,0); \clip (0,0) circle (.23);  \draw (0,0) to[out=40,in=-140] (.5,0);}}
\newcommand{\tpd} {\tikz[scale=.5, thick]{\linebox; \draw[black!30] (0,0) to[out=40,in=-140] (.5,0); \clip (.5,0) circle (.23); \draw (0,0) to[out=40,in=-140] (.5,0);}}
\newcommand{\tpdi}{\tikz[scale=.5, thick]{\linebox; \draw[black!30] (0,0) to[out=40,in=-140] (.5,0); \clip (.5,0) circle (.3);  \draw (0,0) to[out=40,in=-140] (.5,0);}}
\newcommand{\std} {\tikz[scale=.5, thick]{\linebox; \draw[black!30] (0,0) node[black] {\textbf{.}} to[out=40, in=-140] (.5,0) node[black] {\textbf{.}};}}

\title{A First Order Theory of Diagram Chasing}

\author{Assia Mahboubi}{Nantes Université, École Centrale Nantes, CNRS, INRIA, LS2N, UMR 6004, France}{}{0002-0312-5461}{}

\author{Matthieu Piquerez}{Nantes Université, École Centrale Nantes, CNRS, INRIA, LS2N, UMR 6004, France}{}{}{}

\authorrunning{A. Mahboubi and M. Piquerez}

\Copyright{Assia Mahboubi and Matthieu Piquerez}

\ccsdesc[500]{Theory of computation~Logic and verification}

\keywords{Diagram chasing, formal proofs, abelian categories, decidability}

\begin{document}

\maketitle

\begin{abstract}
This paper discusses the formalization of proofs ``by diagram chasing'', a standard technique for proving properties in abelian categories. We discuss how the essence of diagram chases can be captured by a simple many-sorted first-order theory, and we study the models and decidability of this theory. The longer-term motivation of this work is the design of a computer-aided instrument for writing reliable proofs in homological algebra, based on interactive theorem provers.

\end{abstract}

\section{Introduction}\label{sec:intro}

Homological algebra~\cite{zbMATH00758278} attaches and studies a sequence of algebraic objects, typically groups or modules, to a certain space, \eg, a ring or a topological space, in order to better understand the latter. In this field, diagram chasing is a major proof technique, which is usually carried out via a form of diagrammatic reasoning on abelian categories. A diagram can be seen as a functor $F\colon J \rightarrow \mathcal{C}$,  whose domain $J$, the indexing category, is a small category~\cite{riehl}. Diagrams are usually represented as directed multi-graphs, also called  \emph{quivers}, whose vertices are decorated with objects of $\Cat$, and arrows with morphisms. Paths in such graphs thus correspond to chains of composable arrows. Diagrams allow for visualizing the existence of certain morphisms, and to study identities between certain compositions of morphisms. In particular, a diagram \emph{commutes} when any two paths with same source and target lead to identical composite. For instance, the commutativity of the following diagram:
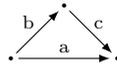
\begin{center}
  \begin{tikzpicture}[inner sep=1pt, -latex, scale=0.7]
  \node (0) at (0,0) {\textbf{.}};
  \node (1) at (1,1) {\textbf{.}};
  \node (2) at (2,0) {\textbf{.}};
  \scriptsize
  \draw (0) -- (1) node[midway, above left] {b};
  \draw (1) -- (2) node[midway, above right] {c};
  \draw (0) -- (2) node[midway, above=1pt] {a};
  \end{tikzpicture}
\end{center}
asserts that morphism $a$ is equal to the composition of morphisms $c$ and $b$, denoted $b \circ c$. Commutativity of diagrams in certain categories can be used to state more involved properties, and diagram chasing essentially consists in establishing the existence, injectivity, surjectivity of certain morphisms, or the exactness of some sequences, using hypotheses of the same nature. The \emph{five lemma} or the \emph{snake lemma} are typical examples of proofs ``by diagram chasing'', also called diagram chases. On paper, diagrams help conveying in a convincing manner proofs otherwise consisting of overly pedestrian chains of equations. The tension between readability and elusiveness may however become a challenge. For instance, diagram chases may rely on non-trivial duality arguments, that is, on the fact that a property about diagrams in any abelian category remains true after reversing all the involved arrows, although the replay of a given proof \emph{mutatis mutandis} cannot be fulfilled in general.

Motivated in part by the second author's experience in writing intricate diagram chases (see for instance~\cite[p.337]{piquerez:tel-03499730}), this work aims at laying the foundations of a computer-aided instrument for writing reliable proofs in homological algebra, based on interactive theorem provers. The present article discusses the design of a formal language for statements of properties amenable to proofs by diagram chasing, according to three objectives. The first is \emph{simplicity and expressivity}: this language should be at the same time simple enough to be implemented in a formal library, and expressive enough to encompass the desired corpus of results. Then, \emph{duality} arguments in proofs shall follow directly from a meta-property of the language. Finally, the corresponding proof system should allow for \emph{effective} proofs of commutativity clauses, that is, proving that the commutativity of some diagram follows from the commutativity of some other diagrams, so that these proofs can eventually be automated.

\begin{definition} \label{def:sig}
We define the many-sorted signatures $\cyc\Sigma$, resp. $\Sigma$, as follows: the sorts of signature $\cyc\Sigma$, resp. $\Sigma$, are finite, \resp acyclic finite, quivers. The symbols of $\cyc\Sigma$, \resp $\Sigma$, consists of one function symbol $\restr_{m: Q' \to Q}$, of arity $Q \to Q'$, per each quiver morphism, \resp each embedding between acyclic quivers, $m$ and one predicate $\com_Q$ on sort $Q$ for each finite, \resp acyclic finite, quiver $Q$.
\end{definition}

The thesis of the present article is that signature $\Sigma$ fulfills the three above objectives. We validate this thesis by giving a first-order theory for diagrams over small and abelian small categories respectively. We state and prove a duality theorem and motivate the choice of $\Sigma$ over the possibly more intuitive $\cyc\Sigma$ by the effectiveness objective. The rest of the article is organized as follows. We first fix some vocabulary and notations in \cref{sec:prelim}, so as in particular to make~\cref{def:sig} precise. Then, \cref{sec:small} introduces a theory for small categories and describes its models, \cref{sec:dual} discusses duality, before \cref{sec:abel} provides an analogue study for abelian categories. Last, we prove in \cref{sec:dec} the decidability, resp. undecidability, of commutativity clauses in $\Sigma$, resp. $\cyc\Sigma$, before concluding in \cref{sec:concl}.

\section{Preliminaries}\label{sec:prelim}

In all what follows, $\N := \{0, 1, \dots \}$ refers to the set of non-negative integers. If if $k\in\N$, then $[k]$ denotes the finite collection $\{0,\dots,k-1\}$. We denote $\card A$ the cardinal of a finite set $A$. We use the notation $\id$ for the identity map.

\subsection{Quivers}\label{ssec:quiv}

\begin{definition}[General quiver, dual]\label{def:quiv}
A \emph{general quiver} $\Q$ is a quadruple $(V_\Q, A_\Q, s_\Q\colon A_\Q \to V_\Q, t_\Q\colon A_\Q \to V_\Q)$ where $V_\Q$ and $A_\Q$ are two sets. The element of $V_\Q$ are called the \emph{vertices} of $\Q$ and the element of $A_\Q$ are called \emph{arrows}. If $a \in A_\Q$, $s_\Q(a)$ is called the \emph{source} of $a$ and $t_\Q(a)$ is called its \emph{target}.
The \emph{dual} of a quiver $\Q$ is the quiver $\Q^\dual \coloneqq (V_\Q,A_\Q,t_\Q,s_\Q)$, which swaps the source and the target maps of $\Q$.
\end{definition}

\begin{definition}[Morphism, embedding, restriction]
A \emph{morphism of quiver} $m\colon \Q \to \Q'$, is the data of two maps $m_V\colon V_\Q \to V_{\Q'}$ and $m_A\colon A_\Q \to A_{\Q'}$ such that $m_A \circ s_\Q = s_{\Q'} \circ m_A$ and $m_A \circ t_\Q = t_{\Q'} \circ m_A$. Such a morphism is called an \emph{embedding of quivers} if moreover both $m_V$ and $m_A$ are injective. In this case we write $m\colon \Q \hookrightarrow \Q'$.

If $A$ is a subset of $A_\Q$, the \emph{(spanning) restriction of\/ $\Q$ to $A$} denoted $\Q|_A$ is the quiver $(V_\Q, A, s_\Q|_A, t_\Q|_A)$. There is a canonical embedding $\Q|_A \hookrightarrow \Q$.
\end{definition}

\medskip

We denote by $\emptyset$ the empty quiver with no vertex and no
arrow, and by $\cyc S$ the set of quivers $Q$ such that $V_Q$ and $A_Q$ are finite subsets of $\N$. In this article,  a \emph{quiver} refers to an element of $\cyc S$. We use a non-cursive $Q$ for elements of $\cyc S$, and a cursive $\Q$ for general quivers.

For the sake of readability, we use drawings to describe
some elements of $\cyc S$, as for instance:
\begin{center}
\begin{tikzpicture}[inner sep=1pt, -latex,scale=.7]
  \node (0) at (0,0) {\textbf{.}};
  \node (1) at (1,0) {\textbf{.}};
  \node (2) at (2,0) {\textbf{.}};
  \scriptsize
  \draw (0) -- (1) node[midway, above] {};
  \draw (0) to[bend right=35] node[midway, above] {} (2);
  \draw (1) to[bend left=30] node[midway, above] {} (2);
  \draw (1) to[bend right=20] node[midway, above] {} (2);
\end{tikzpicture}
\end{center}
For a quiver $Q$ denoted by such a drawing, the convention is that
 $V_Q = [\card{V_Q}]$ and $A_Q = [\card{A_Q}]$ .
From left to right, the drawn vertices
correspond to $0, 1, \dots, \card{V_Q}-1$. Arrows are  then numbered by sorting pairs $(s_Q, t_Q)$ in
increasing lexicographical order, as in:
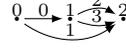
\begin{center}
\begin{tikzpicture}[inner sep=1pt, -latex,scale=.7]
  \node (0) at (0,0) {\textbf{.}};
  \node (1) at (1,0) {\textbf{.}};
  \node (2) at (2,0) {\textbf{.}};
  \scriptsize
  \draw (0) node[above] {0};
  \draw (1) node[above] {1};
  \draw (2) node[above] {2};
  \draw (0) -- (1) node[midway, above] {0};
  \draw (0) to[bend right=35] node[midway, above] {1} (2);
  \draw (1) to[bend left=30] node[midway, above] {2} (2);
  \draw (1) to[bend right=20] node[midway, above] {3} (2);
  \end{tikzpicture}
\end{center}

We also use drawings to denote embeddings. The black part represents the domain of the morphism, the union of black and gray parts represents its codomain. Here is an example of an embedding of the quiver $\Qcobimap[]$ into the quiver drawn above.
\begin{center}
  \begin{tikzpicture}[inner sep=1pt, -latex,scale=.7]
    \node[black!30] (0) at (0,0) {\textbf{.}};
    \node (1) at (1,0) {\textbf{.}};
    \node (2) at (2,0) {\textbf{.}};
    \scriptsize
    \draw[black!30] (0);
    \draw (1);
    \draw (2);
    \draw[black!30] (0) -- (1);
    \draw[black!30] (0) to[bend right=35] (2);
    \draw (1) to[bend left=30] (2);
    \draw (1) to[bend right=20] (2);
  \end{tikzpicture}
\end{center}

\begin{definition}[Path-quiver]\label{def:pathquiv}
  The \emph{path-quiver of length $k$}, denoted $\PQ_k$, is the quiver
  with $k+1$ vertices and $k$ arrows $([k+1], [k], \id, (i \mapsto i + 1))$.
\end{definition}
A path-quiver can be drawn as:
\begin{center} \tikz[baseline={([yshift=-.5ex]current bounding box.center)}, inner sep=1pt, -latex,scale=.7]{
\bfseries \node (0) at (0,0) {.}; \node (1) at (1,0) {.}; \node (2) at (2,0) {.}; \node (3) at (3.6,0) {.}; \node at (2.8,0) {\normalfont $\dots$}; \draw (0) -- (1); \draw (1) -- (2); \draw (3.2,0) -- (3); \draw[-] (2) -- (2.4,0); }
\end{center} with at least one vertex. Such a path-quiver is called \emph{nontrivial} if it has at least two vertices.

{
\renewcommand{\linebox}{\path[use as bounding box] (-.05,-1ex)--(.55,-1ex)}
If $0 \leq k \leq l$ are two integers, we denote by $\sop_{k,l}\colon \PQ_k \hookrightarrow \PQ_l$ the leftmost embedding of $\PQ_k$ into $\PQ_l$, \ie, such that $(\sop_{k,l})_V(0)=0$. If $k$ and $l$ are clear from the context, we draw $\sop_{k,l}$ as $\spd$ if $k\neq0$ and as $\spo$ if $k=0$. Moreover, we denote by $\tp_{k,l}\colon \PQ_k \hookrightarrow \PQ_l$ the rightmost embedding of $\PQ_k$ into $\PQ_l$, \ie, such that $(\tp_{k,l})_V(k) = l$. The corresponding drawings are $\tpd$ and $\tpo$. Moreover, if $P$ is a nontrivial path-quiver, we define $\sot_P\colon \quiver[]{(0,0),(.2,0)}{} \hookrightarrow P$ to be the embedding mapping the first vertex on the leftmost vertex of $P$ and the second vertex on the rightmost vertex of $P$. We denote this embedding $\std$.
}

If $\Q$ is a general quiver, a morphism of the form $p\colon \PQ_k \to \Q$, for some $k$, is called a \emph{path of\/ $\Q$ from $u$ to $v$ of length $k$}, where $u \coloneqq p(0)$ and $v \coloneqq p(k)$. Two paths $p_1\colon P_1 \hookrightarrow \Q$, $p_2\colon P_2 \hookrightarrow \Q$ of $\Q$ \emph{have the same extremities} if $p_1 \circ \sot_{P_1} = p_2 \circ \sot_{P_2}$. We denote by $\BP_{\!\Q}$ the set of pair of paths of $\Q$ having the same extremities. Let $\Q'$ be another general quiver and $m\colon \Q \to \Q'$ be a morphism. Then we define $m_*(p) \coloneqq m  \circ p$.

A general quiver is \emph{acyclic} if any path of this quiver is an embedding. The set of acyclic quivers in $\cyc S$ is denoted by $S$.

\begin{definition}[Free category]\label{def:freecat}
For a general quiver $\Q$, the \emph{free category over $\Q$}, denoted $\free \Q$ is the category with objects $\Ob_{\free\Q} = V_\Q$ whose morphisms $\Hom_{\free\Q}(u,v)$, for two vertices $u$ and $v$ are the paths from $u$ to $v$. The identity map from $u$ to $u$ is the empty path, and the composition is defined as the concatenation of paths.
\end{definition}

Note that a morphism $m$ of quivers induces a functor between the corresponding categories that we denote by $\Phi_m$.  In the other direction, any small category $\Cat$ has an \emph{underlying quiver}.

\subsection{Diagrams}\label{ssec:diag}

We can now introduce diagrams in a category, and a few useful specific examples thereof.
\begin{definition}[Diagram]\label{def:diag}
For any category $\Cat$ and any quiver $Q$, a \emph{diagram in $\Cat$ over $Q$} is a functor from $\free Q$ to $\Cat$.
\end{definition}

Let $P$ be a path-quiver from vertex $u$ to vertex $v$. To a diagram $D\colon \free{P} \to \Cat$ over $P$ one can associate the corresponding composition of morphisms in the category $\Cat$, which is an element of $\Hom_\Cat(D(u), D(v))$. We denote this element $\comp(D)$. By convention, when the path-quiver $P$ is trivial, $\comp(D)$ is the identity map $\id_{D(u)}$.

\begin{definition}[Pullback]
Let $Q, Q'$ be two quivers and let $D$ be a diagram over $Q$. Let $m\colon Q' \to Q$ be a morphism of quivers. We define the \emph{pullback of $D$ by $m$}, denoted $m^*(D)$, as the diagram $D \circ \Phi_m$ over $Q'$.
\end{definition}

\begin{definition}[Commutative diagram]\label{ssec:diagram_commutativity}
For any category $\Cat$ and any quiver $Q$,a  diagram $D$ over $Q$ is \emph{commutative} if $\comp(p_1^*(D))$ and $\comp(p_2^*(D))$ coincide for any two paths $p_1$ and $p_2$ in $Q$ with same extremities, that is:
\begin{equation} \label{eqn:diagram_commutativity}
  \forall (p_1, p_2) \in \BP_{\!Q}, \quad \comp(p_1^*(D)) = \comp(p_2^*(D)).
\end{equation}
\end{definition}

The next lemma allows to reduce the number of distinct diagrams involved in a formula.

\begin{lemma}
  Let $\Cat$ be some category. Consider two diagrams $D_1$ and $D_2$ in $\Cat$ over $Q_1$ and $Q_2$, respectively. If the pullback of $D_1$ by $m_1$ coincides with the pullback of $D_2$ by $m_2$, \ie, $$m^*_1(D_1) = m^*_2(D_2)$$ then, there exists a unique diagram $D'$ over $Q'$ such that $D_1$ (resp. $D_2$) is the pullback of $D'$ by $m'_1$ (resp. $m'_2$), \ie,
  $$ D_1 = m'^*_1(D') \quad \textrm{and} \quad D_2 = m'^*_2(D')$$
\end{lemma}
\begin{proof} Immediate. \end{proof}

\begin{definition}[Pushout] \label{def:pushout}

Consider four quivers $Q, Q_1, Q_2, Q'$ and four maps $m_1\colon Q \to Q_1$, $m_2\colon Q\to Q_2$, $m_1'\colon Q_1 \to Q'$ and $m_2'\colon Q_2 \to Q'$:
\begin{center} \begin{tikzcd}[sep=scriptsize]
  Q \rar{m_1} \dar{m_2} & Q_1 \dar{m'_1} \\
  Q_2 \rar{m'_2} & Q'.
\end{tikzcd} \end{center}

This data is a \emph{pushout configuration} if:
\begin{itemize}
  \item $m'_1 \circ m_1 = m'_2 \circ m_2$,
  \item $Q' = \im(m'_1) \cup \im(m'_2)$, (i.e., $V_{Q'} = \im(m'_{1,V}) \cup \im(m'_{2,V})$ and $A_{Q'} = \im(m'_{1,A}) \cup \im(m'_{2,A})$),
  \item $\im(m'_1 \circ m_1) = \im(m'_1) \cap \im(m'_2)$.
\end{itemize}
\end{definition}
Consider a pushout configuration as in \cref{def:pushout}. The triple $(Q', m_1', m_2')$ is called a \emph{pushout of\/ $(Q, Q_1, Q_2, m_1, m_2)$}. Such a pushout always exists, and any two such pushouts are isomorphic. Moreover, if $m_1$ and $m_2$ are embeddings, then so are $m'_1$ and $m'_2$.

\begin{figure}[b]
  \[ P_1 = \quiver[]{(0,0),(.66,.2),(1.33,.2),(2,0)}{0/1,1/2,2/3} \qquad P_2 = \quiver[]{(0,0),(1,-.2),(2,0)}{0/1,1/2} \qquad Q' = \quiver[]{(0,0),(.66,.2),(1.33,.2),(2,0),(1,-.2)}{0/1,1/2,2/3,0/4,4/3} \]
  \caption{A pushout $Q'$ of two path-quivers $P_1$ and $P_2$ with respect to $\sot_{P_1}$ and $\sot_{P_2}$. \label{fig:pushout_path-quivers}}
\end{figure}

\begin{lemma}
  Let $\Cat$ be some category. Consider two diagrams $D_1$ and $D_2$ in $\Cat$ over $Q_1$ and $Q_2$, respectively. If the pullback of $D_1$ by $m_1$ coincides with the pullback of $D_2$ by $m_2$, \ie,
  \[m^*_1(D_1) = m^*_2(D_2)\]
  then there exists a unique diagram $D'$ over $Q'$ such that $D_1$ (resp. $D_2$) is the pullback of $D'$ by $m'_1$ (resp. $m'_1$), \ie,
\[ D_1 = m'^*_1(D') \quad \textrm{and} \quad D_2 = m'^*_2(D')\]
\end{lemma}
\begin{proof} Immediate. \end{proof}

\subsection{Category relations, path relations and quotient categories} \label{ssec:quotient_category}

We first name relations on the morphisms of a category that are compatible with composition:

\begin{definition}[Category relation]\label{def:catrel}
  A \emph{category relation $r$} on $\Cat$ is by definition the data of an equivalence relation $r_{A,B} \subseteq \Hom(A,B)^2$ for any pair of objects $A$ and $B$ such that, for any objects $A,B,C$ and any morphisms $f,g \in \Hom(A,B)$ and $f',g' \in \Hom(B,C)$,
  \[ f \rel{r} g \ \text{ and } \ f' \rel{r} g' \quad \Longrightarrow \quad f' \circ f \rel{r} g' \circ g, \]
  where we write $h \rel{r} h'$ if $h$ and $h'$ are in relation, \ie,
  $(h, h') \in r_{A,B}$. Such a relation is said \emph{complete} if
  $r_{A,B} = \Hom(A,B)^2$ for any pair of objects $A$ and
  $B$.
\end{definition}

\begin{lemma}[Quotient category]\label{sec:quotient_category}
Given such a category relation, we define the \emph{quotient category $\Cat/r$} given by $\Ob_{\Cat/r} := \Ob_\Cat$ and $\Hom_{\Cat/r}(A,B) = \Hom_{\Cat}(A,B)/r_{A,B}$ is indeed a category for the induced composition.
\end{lemma}
\begin{proof}
  Immediate.
\end{proof}
We now name the relations between the paths of a general quiver induced by the corresponding identities of morphism composites, in the associated free category.

\begin{definition}\label{def:pathrel}
  A \emph{relation between paths with same extremities in $\Q$} is by definition a subset of $\BP_{\!\Q}$. If $r \subseteq \BP_{\!\Q}$ is such a relation then, for $(p,q) \in \BP_{\!\Q}$, we write $p \rel{r} q$ if $(p,q) \in r$. The complete path relation on $\Q$, \ie, $\BP_{\!\Q}$, is denoted $\tot_\Q$. Note that $\BP_{\!\Q} = \bigsqcup_{A,B \in \Ob_{\free{\Q}}} \Hom_{\free{\Q}}(A,B)^2$. Such a relation $r$ is called a \emph{path relation} if it is a category relation on $\free{\Q}$.
\end{definition}
For instance, the equality of compositions in a small category $\Cat$ induces a path relation on the underlying quiver $\Q$. If $r_1, \dots, r_l$ are some relations between paths with same extremities, we denote by $(r_1, \dots, r_l)$ the smallest path relation containing $r_1$, \dots, $r_l$.

Let $\Q'$ be another general quiver and let $m\colon \Q \to \Q'$ be a morphism. If $r \subseteq \BP_{\!\Q}$, we denote by $m_*(r)$ the relation induced by the image by $m$ of $r$ in $\BP_{\!\Q'}$.

\subsection{Many-sorted logic, categorical interpretation}\label{ssec:many}

We first recall a few basic definitions mostly pertaining to many-sorted logic, applied to the signatures introduced by \cref{def:sig}, and we set the corresponding notations.

Let us first fix a countable set $X$, so that for each quiver $Q$ in $\cyc S$ (resp. in $S$), elements of the set $X_Q := X \times \{Q\}$ are the \emph{variables of sort $Q$}. A \emph{term} of sort $Q$ either is a variable of sort $Q$ or has the form $\restr_{m\colon Q' \to Q}(t)$, with $t$ a term of sort $Q$. When possible, we leave the sorts implicit and simplify the notation of symbol $\restr_{m\colon Q' \to Q}(t)$ into $\restr_m$.

We denote the equality symbols by $\eqd$. An \emph{atom} is thus of the form $s \eqd t$ with $s$ and $t$ two terms of the same sort, or of the form $\com_Q(t)$ with $t$ a term of sort $Q$. We consider first-order many-sorted formulas and write the sort of quantifiers as a subscript, \ie, $\exists_Q x_Q, \phi$ and $\forall_Q x_Q, \phi$ where $Q \in \cyc S, x_Q \in X_Q$ and $\phi$ is a formula. In what follows, we however drop sort subscripts when they are clear from the context.

We write $\exists! y, P(y)$ for formula $\bigl(\exists y,\  P(y)\bigr)   \wedge  \bigl(\forall y_1, y_2,\ P(y_1) \wedge P(y_2) \to y_1 \eqd y_2\bigr)$. A formula with free variables $x_1,\dots, x_n$ of respective sorts $Q_1,\dots,Q_n$ is said to be of \emph{arity} $Q_1\times \dots Q_n$. \begin{definition}[Models] Let $Y\subseteq X \times \cyc S$ be a set of variables. An \emph{interpretation} (also called a \emph{model}) \emph{$\M$ of\/ $\cyc\Sigma$ over $Y$} is a map such that
\begin{itemize}
  \item each sort $Q$ is mapped to a domain set, denoted $\M_Q$,
  \item each variable $x \in Y$ of sort $Q$ is mapped to an element $x^\M \in \M_Q$,
  \item each function $\restr_{m\colon Q'\rightarrow Q}$ is mapped to function $\restr^\M_{m\colon Q'\to Q}\colon \M_Q \to \M_{Q'}$,
  \item each predicate $\com_Q$ is mapped to a subset $\com^\M_Q$ of $\M_Q$.
\end{itemize}
An interpretation is extended in the usual way to all formulas whose free variables are in $Y$. An interpretation $\M$ of\/ $\cyc\Sigma$ is an interpretation over the empty set.

We define interpretations (models) of $\Sigma$ in a similar fashion.
\end{definition}

The \emph{evaluation of a formula $\phi$ under an interpretation $\M$} is either true or false depending on the truthiness of the interpretation of the formula. If $\phi$ has no free variable, we write $\M \models \phi$ if $\phi^\M$ is true.

A \emph{theory} is a set of formulas for a certain signature. Remember that signature $\Sigma$ only differs from $\cyc \Sigma$ by restricting the allowed sorts (to the acyclic quivers).  If $\cyc T$ is a theory for $\cyc \Sigma$, we thus define the \emph{restriction of $\cyc T$ to $\Sigma$}, denoted $\cyc T|_\Sigma$, as the subset of formulas of $\cyc T$ which are well-formed with respect to $\Sigma$. A \emph{model of a theory} is a model such that the interpretation of every formula of the theory is true.

The  prototypical models of the signatures introduced in \cref{def:sig} are actually diagrams over a certain category.

\begin{definition}[Categorical interpretation]\label{def:catinterp}
To each small category $\Cat$, we associate an interpretation of $\Sigma$, \resp $\cyc\Sigma$, that we also denote by $\Cat$, as follows.
\begin{itemize}
  \item To each sort $Q$ we associate the set $\Cat_Q$ of diagrams in $\Cat$ over $Q$.
  \item $\restr_m$ is interpreted as the function mapping a diagram $D$ to the diagram $m^*(D)$.
  \item $\com^\Cat_Q$ is the set of commutative diagrams in $\Cat$ over $Q$.
\end{itemize}
We call such an interpretation a \emph{categorical interpretation of\/ $\Sigma$, \resp $\cyc\Sigma$}.
\end{definition}

\section{A theory for diagrams over small categories}\label{sec:small}

This section introduces a theory whose models of can be seen as categorical
interpretations.

\subsection{Axioms}\label{ssec:definition_Tcat}

We now introduce the different axioms of the theory. A formula $F$ with free variables $x_1,\dots,x_n$ is written $F(x_1,\dots,x_n)$ so as to clarify the sorts of each variable in the arity of $F$.

\paragraph*{Existence and uniqueness of the empty diagram}\label{sssec:empty_diagram_uniq}
\[ \EmptyEU\colon\qquad \exists!_\emptyset x, \quad x \eqd x \]

\paragraph*{Compatibility of restrictions} \label{sssec:restr_comp}

For any quivers $Q, Q', Q''$ and morphisms $m\colon Q \to Q'$ and $m'\colon Q' \to Q''$, we define:
\[ \RestrComp_{m,m'}\colon\qquad \forall_{Q''} x'', \quad \restr_m(\restr_{m'}(x'')) \eqd \restr_{m' \circ m}(x''). \]

\paragraph*{Pushout} \label{sssec:pushout}

For any pushout configuration as in \cref{def:pushout}, and using the same notations as this definition, we define the following formulas of arity $Q_1 \times Q_2 \times Q'$:
\[ \Pushout_{m'_1,m'_2}(x_1,x_2,x')\colon\qquad \restr_{m'_1}(x') \eqd x_1 \ \wedge \ \restr_{m'_2}(x') \eqd x_2. \]
\begin{align*}
  \PushoutE_{m_1,m_2,m'_1,m'_2}\colon\qquad \forall x_1, x_2,\quad&  \restr_{m_1}(x_1) \eqd  \restr_{m_2}(x_2) \\
  &\qquad\to \  \exists!\,x', \ \Pushout_{m'_1,m'_2}(x_1,x_2,x').
\end{align*}

\paragraph*{Composition}

The following formula, of arity $\Qmap[] \times \Qmap[] \times \Qmap[]$, describes composite of arrows:
\begin{align*}
  \Comp(x,y,z)\colon\qquad& \exists_{\Qcomp} w,\quad \restr_{\mapicomp}(w) \eqd x \ \wedge \ \restr_{\mapiicomp}(w) \eqd y \\
  &\qquad \wedge \ \restr_{\mapiocomp}(w) \eqd z \ \wedge \ \com(w)
\end{align*}
while the following one ensures the existence of compositions:
\[ \CompE\colon\qquad \forall x,y, \exists z, \quad \Comp(x,y,z). \]

\paragraph*{Equality of nontrivial paths}

For any two nontrivial path-quivers $P_1$ and $P_2$, and $Q'$, $m'_1$, $m'_2$ such that the following diagram forms a pushout configuration (as for instance on~\cref{fig:pushout_path-quivers})
\begin{center}
  \begin{tikzcd}[sep=scriptsize]
  \Qtwodots[] \rar{\sot_{P_1}} \dar{\sot_{P_2}} & P_1 \dar{m'_1} \\
  P_2 \rar{m'_2} & Q'
  \end{tikzcd}
\end{center}
we define the following formula of arity $P_1 \times P_2$:
\begin{align*}
  \EqPath_{P_1,P_2}(x_1,x_2)\colon\qquad
    & \restr_{\std}(x_1) \eqd \restr_{\std}(x_2) \\
    &\qquad \wedge \quad \bigl(\forall x, \ \Pushout_{m'_1, m'_2}(x_1,x_2,x) \to \com(x)\bigr).
\end{align*}

\paragraph*{Identity}

The following formula, of arity $\Qdot[] \times \Qmap[]$, defines the identity map:
\begin{align*}
  \Id(x, y)\colon\qquad
    &\restr_{\dotimap}(y) \eqd x \quad \wedge \quad \forall z, w, \\
    & \ \bigl(\Comp(y,z,w) \to \EqPath(z,w)\bigr) \ \wedge \ \bigl(\Comp(z,y,w) \to \EqPath(z,w)\bigr)
\end{align*}
and the following formula ensures the existence of identity maps.
\[ \IdE\colon\qquad \forall x, \exists y, \quad \Id(x,y). \]

\paragraph*{Equality for general paths}

For any nontrivial path-quiver $P$, we define the following formulas:
\begin{gather*}
  \EqPath_{\Qdot,\Qdot}(x,y)\colon\qquad x \eqd y, \\
  \EqPath_{\Qdot,P}(x,y)\colon\qquad \exists z,\quad \Id(x,z) \ \wedge \ \EqPath_{\Qmap[.5],P}(z,y), \\
  \EqPath_{P,\Qdot}(x,y)\colon\qquad \EqPath_{\Qdot,P}(y,x).
\end{gather*}

Hence, we can see $\EqPath$ as a relation with arity any pair of path-quivers. We ensure this relation to be an equivalence relation by defining for any path-quivers $P_1, P_2$ and $P_3$ the three formulas $\EqPathRefl_{P_1}$, $\EqPathSym_{P_1,P_2}$ and $\EqPathTrans_{P_1,P_2,P_3}$ stating that the relation $\EqPath$ is respectively reflexive, symmetric and transitive.

\medskip

We also make sure to enforce the properties of a category relation. For this purpose, for any four path-quivers $P_1$, $P_1'$, $P_2$, and $P_2'$, of respective length $k_1, k_1', k_2$ and $k_2'$, we define the following formula where bound variables $x_1, x_2, x'_1, x'_2$ respectively have sort $P_1$, $P_2$, $P'_1$ and $P'_2$ and the sort of $x''_i$, $i \in \{1,2\}$, is $\PQ_{k_i+k'_i}$.
\begin{align*}
  &\EqPathConcat_{P_1,P_2,P'_1,P'_2}\colon\qquad \forall x_1, x_2, x'_1, x'_2,\\
    &\quad \EqPath(x_1,x_2) \ \wedge\  \EqPath(x'_1,x'_2) \ \wedge\  \restr_{\tpo}(x1) \eqd \restr_{\spo}(x_1') \\
    &\quad \ \to\quad \forall x''_1, x''_2,\quad \Pushout_{\spd,\tpdi}(x_1,x'_1,x''_1) \ \wedge\ \Pushout_{\spd,\tpdi}(x_2,x'_2,x''_2) \\
    &\hspace{4cm} \to\quad \EqPath(x''_1,x''_2).
\end{align*}

\paragraph*{Commutativity}

We relate commutativity and equality by the following formula:
\[ \ComEq\colon\qquad \forall_{\Qcobimap} x,\quad \com(x)\ \to \ \restr_{\quiverr[.5]{(0,0),(1,0)}{0/1/-30,0/1/30}[1]}(x)\eqd\restr_{\quiverr[.5]{(0,0),(1,0)}{0/1/30,0/1/-30}[1]}(x). \]

For any quiver $Q$, the following formula provides an analogue of the notion of commutativity of diagrams given in \cref{ssec:diagram_commutativity}, \cref{eqn:diagram_commutativity}:
\[ \PathCom_Q\colon\qquad \forall_Q x, \bigwedge_{(p_1, p_2)\in\BP_{\!Q}}\hspace{-1.6em} \EqPath(\restr_{p_1}(x),\restr_{p_2}(x)) \quad \leftrightarrow \quad \com(x). \]
where we recall that $\BP_{\!Q}$ denotes the set of pair of paths of $Q$ having the same extremities.

\begin{definition}
Theory $\cyc\T_\cat$, over signature $\cyc\Sigma$, consists of the following formulas:
\begin{itemize}
  \item $\EmptyEU$, $\CompE$, $\IdE$, $\ComEq$,
  \item $\RestrComp_{m,m'}$ for any pair of maps $m$ and $m'$ as in \cref{sssec:restr_comp},
  \item $\PushoutE_{m_1,m_2,m_1',m_2'}$ for a pushout configuration as in \cref{def:pushout},
  \item $\EqPathRefl_{P_1}, \EqPathSym_{P_1,P_2}, \EqPathTrans_{P_1,P_2,P_3}, \EqPathConcat_{P_1,P_2,P_3,P_4}$ for any quadruple of path-quivers $P_1, P_2, P_3$ and $P_4$,
  \item $\PathCom_Q$ for any quiver $Q$.
\end{itemize}
Theory $\T_\cat$ is defined as the restriction of $\cyc\T_\cat$ to $\Sigma$.
\end{definition}

\subsection{Models} \label{sec:proof_category_theory}

Models of $\T_\cat$, \resp $\cyc\T_\cat$ are in fact exactly what we have called categorical interpretations.
\begin{theorem} \label{thm:category_theory}
  Every categorical interpretation of $\Sigma$, \resp $\cyc\Sigma$, is a model of\/ $\T_\cat$, \resp $\cyc\T_\cat$. Moreover, any model $\M$ of\/ $\T_\cat$, \resp $\cyc\T_\cat$, has an isomorphic categorical interpretation.
\end{theorem}

\begin{proof}
We only prove the theorem for $\T_\cat$, as the proof for $\cyc\T_\cat$ is similar.

If $\Cat$ is a small category, a routine check shows that the associated model $\Cat$ of $\Sigma$ verifies the theory $\T_\cat$. For instance,
\begin{itemize}
  \item the formulas $\PushoutE$ follows from the remark in \cref{def:pushout},
  \item $\IdE$ and $\CompE$ come from the existence of the identity map and the existence of the composition respectively,
  \item for two path-quivers $P_1$, $P_2$ and two diagrams $D_1$ and $D_2$ over them, $\EqPath_{P_1,P_2}(D_1,D_2)$ is the relation $\comp(D_1) = \comp(D_2)$, which is a path relation by \cref{sec:quotient_category},
  \item $\ComEq$ and $\PathCom_Q$ follow from \cref{eqn:diagram_commutativity}.
\end{itemize}

\smallskip
Let us prove the other direction. Let $Q$ be an acyclic quiver. By an abuse of the notations, if $v \in V_Q$, \resp $a \in A_Q$, we also denote by $v$ the corresponding embedding $v\colon \Qdot[] \hookrightarrow Q$, \resp $a\colon \Qmap[] \hookrightarrow Q$. Here is a crucial lemma for the proof of the theorem.

\begin{lemma}[General pushout] \label{lem:pushout}
  Let $\M$ be a model of $\T_\cat$. Let $Q$ be an acyclic quiver, $\beta_V\colon V_Q \to \M_{\Qdot}$ and $\beta_A\colon A_Q \to \M_{\Qmap}$ be two maps. Then the following statements are equivalent.
  \begin{enumerate}
    \item \label{lem:pushout:1} There exists an element $\beta$ of $\M_Q$ such that $\restr_v^\M(\beta) = \beta_V(v)$ and $\restr_a^\M(\beta) = \beta_A(a)$ for any $v \in V_Q$ and any $a \in A_Q$,
    \item \label{lem:pushout:2} For any $a \in A_Q$, $\restr^\M_{\dotimap}(\beta_A(a)) = \beta_V(s_Q(a))$ and $\restr^\M_{\dotiomap}(\beta_A(a)) = \beta_V(t_Q(a))$.
  \end{enumerate}
  Moreover, when both statements hold, the element $\beta$ is unique.
\end{lemma}

\begin{proof}
  The fact that the first point induces the second one follows directly from $\RestrComp$. Hence we focus on the other direction and on the uniqueness.

  \smallskip
  We proceed by induction on the structure of $Q$. If $Q=\emptyset$ is empty, the second point holds trivially. The first point and the uniqueness follows from $\EmptyEU$.

  \smallskip
  Assume first that the lemma holds for some quiver $Q_1$ with no arrow. Let $Q$ be the quiver $Q_1$ with an extra vertex $v_0$. Let $\beta_V$ and $\beta_A$ be two maps as in the statement of the lemma, and $\beta_{1,V}$ the restriction of $\beta_V$ to $Q_{1,V}$. We have the following pushout configuration:
  \begin{center} \begin{tikzcd}[sep=scriptsize]
    \emptyset \rar[hook]{m_1} \dar[hook]{m_2} & Q_1 \dar[hook]{m'_1} \\
    \Qdot[] \rar[hook, "m'_2 = v_0"'] & Q
  \end{tikzcd} \end{center}

  Point~\ref{lem:pushout:2} holds trivially. Let us prove the existence and the uniqueness of $\beta$ satisfying the property of Point~\ref{lem:pushout:1}. By induction, we get a unique $\beta_1 \in \M_{Q_1}$ compatible with $\beta_{1,V}$ and $\beta_A$. Set $\beta_2 \coloneqq \beta_V(v_0)$.  By $\EmptyEU$, $\restr^\M_{m_1}(\beta_1) = \restr^\M_{m_2}(\beta_2)$. Hence we can apply $\PushoutE_{m_1, m_2, m'_1, m'_2}$ to get a unique element $\beta \in \M_Q$ such that $\Pushout^\M(\beta_1,\beta_2,\beta)$. Using $\RestrComp$ and the induction hypothesis, we can check that for $\beta' \in \M_Q$, there is an equivalence between $\Pushout^\M(\beta_1,\beta_2,\beta')$ and $\restr_v^\M(\beta') = \beta_V(v)$ for any $v \in V_Q$. This concludes the induction.

  \smallskip
  Assume more generally that the lemma holds for some acyclic quiver $Q_1$. Let $Q$ be an acyclic quiver obtained from $Q_1$ by adding one arrow $a_0$. Let $\beta_V$ and $\beta_A$ be two maps as before and $\beta_{1,A}$ the restriction of $\beta_A$ to $A_{Q_1}$. Let $m_1\colon \Qtwodots[] \hookrightarrow Q_1$ mapping the first point to $s_Q(a_0)$ and the second point to $t_Q(a_0)$. Let $m_2\coloneqq\sot_{\Qmap}=\twodotsimap[]$. Once again, we get a pushout configuration:
  \begin{center} \begin{tikzcd}[sep=scriptsize]
    \Qtwodots[] \rar[hook]{m_1} \dar[hook]{m_2} & Q_1 \dar[hook]{m'_1} \\
    {\tikz[baseline={([yshift=-.6ex]current bounding box.center)}, inner sep=1pt, -latex]{
      \node (0) at (0,0)  {\tikz{\fill(0,0) circle (.75pt);}};
      \node (1) at (.8,0) {\tikz{\fill(0,0) circle (.75pt);}};
      \draw (0) -- (1);
    }}
  \rar[hook, "m'_2 \coloneqq a_0"'] & Q
  \end{tikzcd} \end{center}

  Assume that Point~\ref{lem:pushout:2} holds. By induction, we get a unique element $\beta_1 \in \M_{Q_1}$ compatible with $\beta_V$ and $\beta_{1,A}$. Set $\beta_2 \coloneqq \beta_A(a_0)$. We have already proven the lemma for the quiver $\Qtwodots[]$. Hence we deduce that
  \[ \restr^\M_{m_1}(\beta_1) = \restr^\M_{m_2}(\beta_2). \]
  We can apply $\PushoutE$ as before to get Point~\ref{lem:pushout:1} as well as the uniqueness part. This concludes the proof of the lemma.
\end{proof}

We now continue the proof of \cref{thm:category_theory}. Let $\M$ be a model of $\T_\cat$. We define the general quiver $\Q$ associated to $\M$ by
\[ V_\Q \coloneqq \M_{\Qdot}, \quad A_\Q \coloneqq \M_{\Qmap}, \quad s_\Q \coloneqq \restr_{\dotimap}^\M\text{ and} \quad t_\Q \coloneqq \restr_{\dotiomap}^\M. \]
Set $\~\Cat \coloneqq \free\Q$. Thanks to \cref{lem:pushout}, to each acyclic quiver $Q$ and to each element $\beta \in \M_Q$, we can associate a unique diagram $\~\Psi(\beta)$ in $\~C$ verifying:
\begin{itemize}
  \item for any $v\in V_Q$, $\~\Psi(\beta)(v) = \restr_v^\M(\beta)$.
  \item for any $a\in A_Q$, $\~\Psi(\beta)(a)$ is the path of length one with arrow $\restr_a^\M(\beta)$.
\end{itemize}
The image of $\~\Psi$ is exactly the set of diagrams whose morphisms are paths of lengths one.

There is another important map. Let $A$ and $B$ be two objects of $\~\Cat$, and let $p \in \Hom(A,B)$. Recall that $p$ is just a path from $A$ to $B$ in $Q$. Let $k$ be the length of $p$. By \cref{lem:pushout}, there exists a unique element $\Theta(p) \in \M_{\PQ_k}$ such that
\begin{itemize}
  \item for each $v \in V_{\PQ_k}$, $\restr_v^\M(\Theta(p)) = p(v)$,
  \item for each $a \in A_{\PQ_k}$, $\restr_a^\M(\Theta(p)) = p(a)$.
\end{itemize}
Relation $\EqPath^\M$ thus induces a relation $r$ on morphisms of $\~\Cat$. Moreover, $\EqPathRefl$, $\EqPathSym$, $\EqPathTrans$ and $\EqPathConcat$, together with \cref{lem:pushout}, make $r$ a category relation. We can hence define the category $\Cat \coloneqq \~\Cat/r$. Now $\~\Psi$ induces a map $\Psi\colon \M_Q \to \Cat_Q$ for any quiver $Q$, and we claim that $\Psi$ induces a model isomorphism between $\M$ and $\Cat$.

Let $Q$ be any acyclic quiver. We first prove that $\Psi$ is injective. Let $\beta,\gamma \in \M_Q$ such that $\Psi(\beta) = \Psi(\gamma)$. For any vertex $v\in V_Q$, $\~\Psi(\beta)(v) = \~\Psi(\gamma)(v)$, \ie, $\restr_v^\M(\beta) = \restr_v^\M(\gamma)$. Let $a$ be an arrow of $Q$. Then we have the relation $\~\Psi(\beta)(a) \rel{r} \~\Psi(\gamma)(a)$. By  definition of $\EqPath_{\Qmap,\Qmap}$, we validate the premise of $\ComEq$, and thus the equality $\~\Psi(\beta)(a) = \~\Psi(\gamma)(a)$, \ie, $\restr_a^\M(\beta) = \restr_a^\M(\gamma)$. By the uniqueness part of \cref{lem:pushout}, we get $\beta = \gamma$.

\smallskip
We now consider the surjectivity of $\Psi$. It suffices to prove that any morphism $p \in \Hom_{\~\Cat}(A,B)$ is in relation via $r$ to a path of length one. Indeed, in such a case, for any diagram $\~D$ in $\~\Cat$ over $Q$, one can find another diagram $\~D'$ over $Q$ whose morphisms are path of size one and such that any morphism of $\~D$ is in relation with the corresponding morphism of $\~D'$. Hence the induced diagrams in $\Cat$ are equal. Moreover, $\~D'$ is in the image of $\~\Psi$, and we would get the surjectivity.

Let $P$ be a path-quiver and let $\beta \in \M_P$. We have to find an element $\gamma \in \M_{\Qmap}$ such that $\EqPath^\M_{P,\Qmap}(\beta,\gamma)$. If $P$ has length one, this is trivial. If $P$ has length zero, then by $\IdE$, there exists $\gamma$ such that $\Id^\M(\beta,\gamma)$. Moreover, by the definition $\EqPath_{\Qdot,\Qmap}$ and using the reflexivity of $\EqPath$, we get $\EqPath^\M_{\Qdot,\Qmap}(\beta,\gamma)$. If $P$ has length two, then by $\CompE$, we can find an element $\gamma$ such that
\[ \Comp(\restr^\M_{\mapibimap}(\beta),\restr^\M_{\mapiobimap}(\beta),\gamma). \]
The commutativity of the triangle induces $\EqPath^\M_{\Qbimap,\Qmap}(\beta,\gamma)$.

{
\renewcommand{\linebox}{\path[use as bounding box] (-.05,-1ex)--(.55,-1ex)}
For $P$ with length $k>2$, we work by induction. Using $\EqPathTrans$, it suffices to find $\gamma$ over $\PQ_{k-1}$ such that $\EqPath^\M_{P,\PQ_{k-1}}(\beta,\gamma)$. To do so, we see $P$ as the pushout of $\PQ_2$ and $\PQ_{k-2}$ along $m'_1 = \spd\colon \PQ_2 \to P$ and $m'_2 = \tpdi\colon \PQ_{k-2} \to P$. By the case $k=2$, we can find $\gamma_1 \in \M_{\Qmap}$ such that $\EqPath(\restr^\M_{m'_1}(\beta),\gamma_1)$. We set $\gamma_2 = \restr^\M_{m'_2}(\beta)$. In particular we have $\EqPath(\gamma_2,\gamma_2)$. By $\EqPathConcat$, we get $\EqPath(\beta, \gamma)$ where $\gamma \in \M_{\PQ_{k-1}}$ is such that $\Pushout_{m'_1,m'_2}(\gamma_1,\gamma_2,\gamma)$, and the result follows.
}

\smallskip
We have shown that $\Psi$ induces a bijection between the corresponding domains. In order to conclude the proof, it remains to prove that $\Psi$ commutes with $\restr$ and with $\com$. The commutativity with $\restr$ follows from the definition of $\Psi$ and the formulas $\RestrComp$.

The compatibility with the predicate $\com$ can be reduced to the compatibility of $\EqPath$ via the formula $\PathCom$. Let $P_1$ and $P_2$ be two path-quivers and $\beta_1 \in \M_{P_1}$ and $\beta_2 \in \M_{P_2}$. These elements correspond to paths $p_1$ and $p_2$ in $\Q$. Using the definitions of the different elements, we get the following chain of equivalences:
\begin{align*}
  &\EqPath^\Cat(\Psi(\beta_1), \Psi(\beta_2)) \Leftrightarrow \comp_\Cat(\Psi(\beta_1)) = \comp_\Cat(\Psi(\beta_2)) \\
  &\hspace{3.1cm} \Leftrightarrow \comp_{\~\Cat}(\~\Psi(\beta_1)) \rel{r} \comp_{\~\Cat}(\~\Psi(\beta_2)) \Leftrightarrow p_1 \rel{r} p_2 \Leftrightarrow \EqPath(\beta_1, \beta_2).
\end{align*}
This concludes the proof of the theorem.
\end{proof}

\section{Duality}\label{sec:dual}
Signatures of \cref{def:sig} are tailored to enforce a built-in, therefore easy to prove, duality principle, which we now make precise. Recall from \cref{def:quiv} that duality is an involution on quivers, and also on acyclic quivers. We define the dual of a formula over $\cyc\Sigma$ as follows:
\begin{itemize}
  \item if $m\colon Q \to Q'$ is a morphism, then $m^\dual\colon Q^\dual \to Q'^\dual$ is defined by $m^\dual_V=m_V$ and $m^\dual_A=m_A$.
  \item if $x=(\x,Q)$ is a variable in $X \times \cyc S$ then $x^\dual \coloneqq (\x,Q^\dual)$,
  \item $(\restr_m(x))^\dual \coloneqq \restr_{m^\dual}(x^\dual)$ and $(\com_Q(x))^\dual \coloneqq \com_{Q^\dual}(x^\dual)$,
  \item $(x \eqd y)^\dual \coloneqq x^\dual \eqd y^\dual$,
  \item $(\forall_Q x, \phi)^\dual \coloneqq \forall_{Q^\dual}x^\dual,\phi^\dual$ and $(\exists_Q x, \phi)^\dual \coloneqq \exists_{Q^\dual}x^\dual, \phi^\dual$,
  \item  $(\phi \wedge \psi)^\dual \coloneqq \phi^\dual \wedge \psi^\dual$, etc.
\end{itemize}
For $Y$ a set of variables, $Y^\star$ denotes
$\{ x^\star \mid x \in Y\}$. For any theory $\T$, $\T^\dual$ denotes $\{\phi^\dual \mid \phi\in\T\}$.

For $\M$ an interpretation of $\cyc\Sigma$ over a set of variables $Y$, we define its dual model $\M^\dual$ as:
\begin{itemize}
  \item $\M^\dual_Q \coloneqq \M_{Q^\dual}$,
  \item for $x \in Y^\dual$, $x^{\M^\dual} \coloneqq (x^\dual)^\M$.
  \item $\restr^{\M^\dual}_m \coloneqq \restr^{\M}_{m^\dual}$ and $\com^{\M^\dual}_Q \coloneqq \com^\M_{Q^\dual}$,
\end{itemize}

The duality involution also restricts to formulas, theories and models over $\Sigma$.

\begin{example} \label{ex:dual_category}
  If $\Cat$ is a small category, then the dual interpretation $\Cat^\dual$ is isomorphic to the model of the dual category, both with respect to $\cyc\Sigma$ and to $\Sigma$.
\end{example}

\begin{theorem}[Duality theorem]
  Let $\phi$ be a formula with free variables included in $Y \subseteq X \times S$, \resp in $Y \subseteq X \times \cyc S$, and let $\M$ be a model of\/ $\Sigma$, \resp of $\cyc\Sigma$. Then
  \[ \M \models \phi \quad \Longleftrightarrow \quad \M^\dual \models \phi^\dual. \]
\end{theorem}

\begin{proof}
  More generally, a formula is provable if and only if its dual is provable.
\end{proof}

\begin{remark} The duality principle has some useful direct consequences:
\begin{itemize}
  \item If $\phi$ is a valid, \resp satisfiable, \resp unsatisfiable, formula, so is $\phi^\dual$.
  \item Let $\T$ be theory such that any model of $\T$ verifies $\T^\dual$. If $\phi$ is a valid, \resp satisfiable, \resp unsatisfiable, formula among models of $\T$, so is $\phi^\dual$.
  \item We have the following reciprocal. Let $\T$ be a theory such that any model $\M$ of $\T$ verifies that $\M^\dual \models \T$, then every model of $\T$ verifies $\T^\dual$.
\end{itemize}
\end{remark}

The following fact follows directly from this last point, \cref{thm:category_theory} and \cref{ex:dual_category}.

\begin{proposition} \label{prop:Tcat_autodual}
  Models of\/ $\T_\cat$ verify $\T_\cat^\dual$, and models of\/ $\cyc\T_\cat$ verify $\cyc\T_\cat^\dual$.
\end{proposition}

\section{A theory for diagrams over abelian categories}\label{sec:abel}

We now introduce a theory whose models are diagrams in small abelian categories. We rely on the set of axioms given by Freyd in \cite{Fre64}. This reference is particularly well-suited for our purpose. Indeed, the author does not impose the homomorphisms between any two objects of an abelian category to form a group, but this fact rather follows from the axioms.

\smallskip
Let us first introduce common notions of category theory in our logic. Here is a formula of arity $\Qmap[]$ which corresponds to monicity of a map in a category.
\begin{align*}
  \Mono(x):\qquad& \forall_{\Qmono} y,\quad \restr_{\mapioomono}(y) \eqd x \\
  &\qquad\wedge \ \com(\restr_\compimono(y)) \ \wedge \ \com(\restr_\compiomono(y)) \\
  &\qquad\quad\to \ \com(\restr_\cobimapimono(y)).
\end{align*}
The dual formula is called $\Epi$.

\smallskip
Let $Q$ be a quiver. We define the \emph{cone of $Q$} as the quiver
\[ \cone(Q) \coloneqq (V_Q \sqcup \{v_0\}$, $A_Q \sqcup \{a_v \mid v \in V_Q\}, s_{\cone(Q)}, t_{\cone(Q)}), \]
where $s_{\cone(Q)}$ and $t_{\cone(Q)}$ are extensions of $s_Q$ and $t_Q$ by $s_Q(a_v) = v_0$ and $t_Q(a_v) = v$. We define by $i_Q\colon Q \hookrightarrow \cone(Q)$ the corresponding embedding. If $m\colon Q \to Q'$ is a morphism of quivers, we get a canonical morphism $\cone(m)\colon \cone(Q) \to \cone(Q')$.

Abusing the notations, if $a$ is an arrow of $Q$, we also denote by $a\colon \Qmap[] \hookrightarrow Q$ the corresponding morphism. We then introduce the usual notion of cones of diagrams by the following formula of arity $Q \times \cone(Q)$.
\[ \Cone_Q(x, y)\colon\qquad \restr_{i_Q}(y) \eqd x \quad \wedge \quad \bigwedge_{a \in A_Q} \com(\restr_{\cone(a)}(y)). \]
Here is the notion of limit.
\begin{align*}
  \Limit_Q(x, y)\colon\qquad& \Cone_Q(x, y) \quad \wedge\\
  &\qquad \forall z,\quad \Cone_Q(x, z) \ \to \ \exists! w, \ \Cone(y, w) \wedge \restr_{\cone(i_Q)}(w) \eqd z.
\end{align*}
We also introduce the dual notion $\Colimit_Q \coloneqq \Limit_{Q^\dual}^\dual$.

The introduction of monos, epis, limits and colimits allows to state the axioms of abelian category given by Freyd~\cite{Fre64}. First, we define zero objects and kernels as follows.
\begin{align*}
  \Zero(x)\colon\qquad& \forall y,\quad \Limit_\emptyset(y,x) \ \wedge \ \Colimit_\emptyset(y,x), \\
  \Ker(x,y)\colon\qquad& \exists_{\Qker} z,
    \qquad \restr_{\mapiooker}(z) \eqd x \quad \wedge \quad  \restr_{\mapiker}(z) \eqd y \\
    & \qquad \wedge \quad \Zero(\restr_{\dotiiker}(z)) \quad \wedge \quad \Limit(\restr_{\veeiker}(z), z).
\end{align*}
We also define $\Coker \coloneqq \Ker^\dual$.

\medskip

We define the category $\T_\ab$ as the extension of $\T_\cat$ by the following formulas.
\begin{align*}
  \ZeroE\colon&\qquad \exists_{\Qdot} x, \quad \Zero(x), \\
  \ProductE\colon&\qquad \forall x, \exists y, \quad \Limit_{\Qtwodots}(x,y), \\
  \CoproductE\colon&\qquad \ProductE^\dual, \\
  \KerE\colon&\qquad \forall x, \exists y, \quad \Ker(x,y), \\
  \CokerE\colon&\qquad \KerE^\dual, \\
  \MonoNormal\colon&\qquad \forall x, \quad \Mono(x) \ \to \  \exists y, \ \Ker(y,x), \\
  \EpiNormal\colon&\qquad \MonoNormal^\dual.
\end{align*}

\smallskip
The following theorem states that $\T_\ab$ is a theory for diagrams over abelian categories.

\begin{theorem} \label{thm:abelian_category_theory}
  The categorical interpretation induced by any small abelian category is a model of\/ $\T_\ab$. Conversely, any model of\/ $\T_\ab$ is isomorphic to the categorical interpretation associated to some small abelian category.
\end{theorem}

\begin{proof}
  This follows from \cref{thm:category_theory} and from \cite[Chapter 2]{Fre64}.
\end{proof}

\begin{proposition}
  The theory $\T_\ab$ implies its dual $\T_\ab^\dual$.
\end{proposition}

\begin{proof}
  The theory $\T_\cat$ implies its dual by \cref{prop:Tcat_autodual}. Moreover, $\ZeroE$ clearly implies its dual. Finally, for the other axioms we added, we also added their dual.
\end{proof}

\section{Decidability of the commerge problem}\label{sec:dec}

In this section, we use the notations of \cref{ssec:quotient_category}. Let $Q$ be a quiver, $k \in \N$ and, for each $i \in [k]$, let $Q_i$ be a quiver and $m_i\colon Q_i \to Q$ be a morphism. We define the following formula:
\[ \Commerge_{m_0,\dots,m_{k-1}}\colon\qquad \forall_Q x, \qquad \bigwedge_{i=0}^{k-1} \com(\restr_{m_i}(x)) \quad \to \quad \com(x). \]

\begin{definition}
  Notations as above, the \emph{acyclic, \resp cyclic, commerge problem for morphisms, \resp embeddings, $m_0, \dots, m_{k-1}$ and for a theory $\T$} is the problem of deciding the validity of\/ $\Commerge_{m_0,\dots,m_{k-1}}$ among models of\/ $\Sigma$, \resp $\cyc\Sigma$, verifying the theory $\T$.
\end{definition}

We recall that a \emph{thin category} is a category with at most one morphism between any pair of objects.  Let $\tot_{Q_i} = \BP_{\!Q_i}$ be the complete path relation on $Q_i$. Set $r_i \coloneqq m_{i\,*}(\tot_{Q_i})$ for $i \in [k]$. Recall that $(r_i)_{i \in [k]}$ is the smallest path relation containing the $r_i$ for all $i \in [k]$.

\begin{lemma} \label{lem:commerge_triviality}
  Notation as above, the formula $\Commerge_{m_0,\dots,m_{k-1}}$ is valid among model of $\T_\cat$, \resp $\cyc\T_\cat$, if and only if $\free{Q}/(r_i)_{i \in [k]}$ is a thin category.
\end{lemma}

\begin{proof}
  Set $\Cat \coloneqq \free{Q}/(r_i)_{i \in [k]}$. It is a model of $\T_\cat$, \resp $\cyc\T_\cat$. Moreover, the canonical diagram $D\colon \free{Q} \to \Cat$ verifies the premise of $\Commerge_{m_0, \dots, m_{k-1}}$. If $\Cat$ is not thin, then there ase two paths $p$ and $q$ in $\free{Q}$ with the same extremities which are not in relation. Then $\comp(p^*(D))$ is the class of $p$ in the quotient, which is different of the class of $q$, that is of $\comp(q^*(D))$. Hence $D$ is not commutative.

  For the other direction, by \cref{thm:category_theory}, it suffices to study diagrams in small categories. It is easy to check that any diagram $D'$ over $Q$ in a category $\Cat'$ which verifies the condition of $\Commerge_{m_0, \dots, m_{k-1}}$ factors through $D$, \ie, $D' = \Psi \circ D$ for some functor $\Psi\colon \Cat \to \Cat'$. If $\Cat$ is thin, then for any two paths $p$ and $q$ with same extremities in $Q$,
  \[ \comp(p^*(D')) = \Psi(\comp(p^*(D))) = \Psi(\comp(q^*(D))) = \comp(q^*(D')). \]
  Hence $\Commerge_{m_0, \dots, m_{k-1}}$ is valid.
\end{proof}

\begin{theorem}\label{thm:dec}
  The acyclic commerge problem for $\T_\cat$ is decidable for any tuple of embeddings.
\end{theorem}


\begin{proof}
  By \cref{lem:commerge_triviality}, it suffices to decide if $\free{Q}/(r_i)_{i \in [k]}$. Since $\free{Q}$ and the $\free{Q_i}$, $i\in[k]$ are finite, we can compute relation $(m_{i\,*}(\tot_{Q_i}) \mid i \in [k])$ and check it is complete.
\end{proof}

\begin{proposition} \label{prop:undecidable_commerge}
  There exists a tuple of morphisms for which the cyclic commerge problem for $\cyc\T_\cat$ is undecidable.
\end{proposition}

\begin{proof}
  We proceed by reduction to an undecidability result, due to Markov~\cite{zbMATH03065841}. For $B$ an arbitrary finite set and $\langle B \rangle$ the associated free monoid, let $M$ be the finitely presentable monoid $\langle B \rangle/R$, for $R$ a finite subset of $\langle B \rangle$. The triviality of $M$ is undecidable.

  { \renewcommand{\linebox}{\path[use as bounding box] (-.05,-1ex)--(.55,-1ex)}
  Let $B$, $R$ and $M$ as above. Let $Q = (\{v\}, B \sqcup \{e\}, s_Q, t_Q)$ be a quiver with one vertex and loops labeled by elements of $B$ plus one loop $e$. To each element $\rho \in R$ corresponds a path $p_\rho\colon \PQ_{k_\rho} \to Q$, for some $k_\rho \in \N$. Let $Q_\rho$ be a pushout of the morphisms $\std\colon \Qtwodots[] \hookrightarrow \PQ_{k_\rho}$ and $\twodotsimap[]$. Let $m_\rho$ be the extension to $Q_\rho$ of $p_\rho$ obtained by mapping the new arrow onto $e$. Also set $m_e\colon \Qloop \to Q$ which maps the loop on $e$. We claim that the cyclic commerge problem for the $m_\rho$ and $m_e$ is undecidable. Indeed, $\free{Q}/\bigl(m_e^*(\tot_{\Qloop}), (m_\rho^*(\tot_{Q_\rho}))_{\rho\in R}\bigr)$ is the category associated to the monoid $M = \langle B \rangle/R$. Hence this category is thin if and only if the monoid is trivial. Markov Theorem~\cite{zbMATH03065841} and \cref{lem:commerge_triviality} conclude the proof. }
\end{proof}

We strengthen the previous proposition to the case of embeddings.

\begin{theorem} \label{thm:undecidable_commerge_embeddings}
  There exists a tuple of embeddings for which the cyclic commerge problem for $\cyc\T_\cat$ is undecidable.
\end{theorem}

\begin{proof}
  Let $B$, $R$, $M$ as in the proof of \cref{prop:undecidable_commerge}. Let $Q$ be the quiver $(\{v\}, B, s_Q, t_Q)$. If $k \geq 2$, we define the quiver $\cyc Q^k$ as
  \begin{gather*}
    \cyc Q^k \coloneqq \Big(\{ v_i \mid i \in [k] \}, \{ b_{i,j} \mid b \in B, 0 \leq i < j < k \} \sqcup \{ e_{i,j} \mid i,j \in [k] \}, s_{\cyc Q^k}, t_{\cyc Q^k} \Big) \text{ where} \\
    s_{\cyc Q^k}(b_{i,j}) = v_i, \quad s_{\cyc Q^k}(e_{i,j}) = v_i, \quad t_{\cyc Q^k}(b_{i,j}) = v_j, \quad t_{\cyc Q^k}(e_{i,j}) = v_j.
  \end{gather*}

  We have a projection $\pi: \free{\cyc Q^k} \to \free{Q}$, which maps $b_{i,j}$ on $b$ and $e_{i,j}$ on $\id_v$, and a section $\iota: \free{Q} \to \free{\cyc Q^k}$ defined by mapping $v$ onto $v_0$ and $b$ onto $b_{0,k-1} \circ e_{k-1,0}$, where, as usual, we denote in a same way an arrow and the corresponding path of length one.

  For $A$ any subset of $A_{\cyc Q^k}$, let $m_A\colon \cyc Q^k|_A \hookrightarrow \cyc Q^k$ be the canonical embedding, and let $r_A \coloneqq m_{A\,*}(\tot_{\cyc Q^k|_A})$. Set $r' \coloneqq ((r_A)_{A\in\A})$ for $\A \subset 2^{A_{\cyc Q^k}}$ defined as the set containing
  \begin{itemize}
    \item $A_e \coloneqq \{e_{i,j} \mid i,j\in[k]\}$,
    \item for $i<j$ and $b\in B$,
    \[ A_{b,i,j} \coloneqq \{\underbrace{e_{0,i}}_{\text{if $i\neq0$}}, \ b_{i,j} \ , \underbrace{e_{j,k-1}}_\text{if $j\neq k-1$}, \ b_{0,k-1}\}. \]
  \end{itemize}
  We claim that $\pi$ and $\iota$ induce an equivalence of category between $\free{Q}$ and $\free{\cyc Q^k}/r'$. From the definition of $A_e$, for any $i,j,l \in [k]$, we have $e_{i,j} \circ e_{j,l} \rel{r'} e_{i,l}$ and $e_{i,i} \rel{r'} \id_i$. Now the definition of $A_{b,i,j}$, for $b\in B$ and $0 \leq i < j < k$, induces that $e_{0,i} \circ b_{i,j} \circ e_{j,k-1} \rel{r'} b_{0,k-1}$. These relations generate all $r'$, and they become equalities by applying the projection. Hence $\pi_*\colon \free{\cyc Q^k}/r' \to \free{Q}$ is well-defined. Clearly $\pi\circ\iota$ is identity. Concerning the other direction, for $b\in B$ and $0\leq i < j < k$, we have
  \[ \iota\circ\pi(b_{i,j}) = b_{0,k-1} \circ e_{k-1,0} \rel{r'} e_{0,i} \circ b_{i,j} \circ e_{j,k-1} \circ e_{k-1,0} \rel{r'} e_{0,i} \circ b_{i,j} \circ e_{j,0}. \]
  Hence we get a natural transformation $\eta$ between the identity functor and $\iota \circ \pi_*$ by setting $\eta_i \coloneqq e_{i,0} \in \Hom_{\free{\cyc Q^k}/r'}(i, \iota\circ\pi(i) = 0)$. Since $e_{i,0}$ is an isomorphism, we conclude that there is an equivalence of category between $\free{Q}$ and $\free{\cyc Q^k}/r'$.

  \smallskip
  Recall that $M$ is the monoid $\langle B \rangle/R$. Assume that $k$ is greater than the longest word in $R$. To any word $\rho = b^1b^2\dots b^l$ in $R$ corresponds a subset
  \[ A_\rho \coloneqq \{e_{0,l}, b^1_{0,1}, b^2_{1,2}, \dots, b^l_{l-1,l}\} \subseteq A_{\cyc Q^k}. \]

  Let $\A' \coloneqq \A \cup \{A_\rho \mid \rho \in R\}$. We claim that $\Commerge_{(m_A)_{A\in\A'}}$ is undecidable. Indeed, $\cyc Q^k/((r_A)_{A\in\A'})$ is equivalent as a category to $Q/(\pi_*(r_{A_\rho}))_{\rho \in R}$ which is the category of the monoid $M$. Once again, we conclude using \cref{lem:commerge_triviality} and Markov Theorem.
\end{proof}

\begin{theorem}
  The theory $\T_\cat$ contains undecidable formulas.
\end{theorem}

\begin{proof}
  Let $M$, $B$, $R$ and $Q$ as in the proof of \cref{thm:undecidable_commerge_embeddings}. Let $k$ be the size of the longest word in $R$ plus one. We define the quiver $Q^k$ as
  \[ Q^k \coloneqq \Big(\{v_i \mid i \in [k] \}, \{ b_{i,j} \mid b \in B, 0 \leq i < j < k \} \sqcup \{ e_{i,j} \mid 0\leq i < j < k \}, s_{Q^k}, t_{Q^k} \Big) \]
  where
  \[ s_{Q^k}(b_{i,j}) = v_i, \quad s_{Q^k}(e_{i,j}) = v_i, \quad t_{Q^k}(b_{i,j}) = v_j, \quad t_{Q^k}(e_{i,j}) = v_j. \]
  Note that $Q^k$ is acyclic. Let
  \[ \A'' \coloneqq \{ A_{b,i,j} \mid b \in B, 0 \leq i < j < k \} \ \cup \ \{ A_\rho \mid \rho \in R \}, \]
  where $A_{b,i,j}$ and $A_\rho$ are defined as in the proof of \cref{thm:undecidable_commerge_embeddings}. For $A \in \A''$, let $m_A\colon Q^k|_A \hookrightarrow Q^k$ be the corresponding embedding.
  Consider the formula
  \begin{align*}
    \CommergeWithId_{M}\colon \qquad& \forall_{Q^k} x, \quad \bigwedge_{0\leq i < j < k} \Id(\restr_{v_i}(x), \restr_{e_{i,j}}(x)) \\
    & \qquad\qquad \ \wedge \ \bigwedge_{A \in \A''} \com(\restr_A(x)) \quad \to \quad \com(x).
  \end{align*}
  The projection $\pi\colon Q^k \to Q$ induces a pullback $\pi^*$ between diagrams on $Q$ and that on $Q^k$. It is easy to check that $\pi^*$ induces a bijection between the diagrams on $Q$, verifying the condition of the commerge problem described in the proof of \cref{prop:undecidable_commerge}, and the diagrams on $Q^k$ verifying the condition of $\CommergeWithId_{M}$. Moreover the bijection preserves commutativity. Hence we conclude the proof as for \cref{prop:undecidable_commerge}.
\end{proof}

\section{Conclusion}\label{sec:concl}

We have shown that the many-sorted signature $\Sigma$ is expressive enough to formulate a theory $\T_\cat$ and its extension $\T_\ab$, whose models are exactly and respectively diagrams in small categories, and that in small abelian categories. Restricting sorts to acyclic quivers makes the commerge problem for $\T_\cat$ decidable. A companion file~\cite{file} to this submission illustrates how to implement a deep embedding of formulas of $\Sigma$ using the \Coq proof assistant~\cite{the_coq_development_team_2023_8161141}; its content should be easy to transpose to other proof systems. Theorem \code{duality_theorem_with_theory} shall bring a duality principle, \ie, that a formula of the language is valid if and only if its dual is valid, to any formalized definition of abelian categories. \cref{thm:dec} results in a complete decision procedure for commutativity clauses. The optimizations that make it work on concrete examples however go beyond the scope of the present article.

Similar concerns have motivated the implementation of the accomplished \Globular proof assistant~\cite{DBLP:journals/lmcs/BarKV18}, for higher-dimensional category theory. The closest related work we are aware of yet seem unpublished at the time of writing. Lafont's categorical diagram editor~\cite{lafont}, based on the Unimath library~\cite{DBLP:conf/cpp/Voevodsky11} and Barras and Chabassier's graphical interface for diagrammatic proofs~\cite{barras-chabassier} both provide a graphical interface for generating \Coq proof scripts and visualizing \Coq goals as diagrams. No specific automation is however provided. Himmel~\cite{himmel} describes a formalization of abelian categories in \Lean~\cite{DBLP:conf/cade/MouraKADR15}, including proofs of the five lemma and of the snake lemma, and proof (semi-)automation tied to this specific formalization. Duality arguments are not addressed. Monbru~\cite{monbru} also discusses automation issues in diagram chases, and provides heuristics for generating them automatically, albeit expressed in a pseudo-language.
\looseness=-1

\bibliography{biblio}
\end{document}